\documentclass[12pt,a4paper,reqno]{article}
\usepackage{graphics}
\usepackage{epsfig}

\usepackage{amsmath,amsthm,amssymb}
\newtheorem{theorem}{Theorem}[section]
\newtheorem{corollary}[theorem]{Corollary}
\newtheorem{lemma}[theorem]{Lemma}
\newtheorem{proposition}[theorem]{Proposition}

\theoremstyle{definition}

\theoremstyle{remark}

\numberwithin{equation}{section}

\begin{document}
\title{On Various Parameters of  $\mathbb{Z}_q$-Simplex codes for an even
integer q
 \thanks{The first author would like to gratefully acknowledge the UGC-RGNF[Rajiv Gandhi National Fellowship], New Delhi for providing fellowship
 and the second author was supported by a grant(SR/S4/MS:588/09) 
 for the Department of Science and Technology, New Delhi.}}

\author{P. Chella Pandian \\
Research Scholar,\\
Department of Mathematics,\\ 
School of Mathematical Sciences\\
Bharathidasan University\\
Tiruchirappalli Tamil Nadu 620 024 India\\
{Email: chellapandianpc@gmail.com} 
\bigskip\\
C.~Durairajan\\
Assistant Professor,\\
Department of Mathematics,\\ 
School of Mathematical Sciences\\
Bharathidasan University\\
Tiruchirappalli Tamil Nadu 620 024 India\\
Email: durairajanc@gmail.com
\hfill \\
\hfill \\
\hfill \\
\hfill \\
{\bf Proposed running head:} On Various Parameters of  $\mathbb{Z}_q$-Simplex codes\\ for an even
integer q}
\date{}
\maketitle

\newpage

\vspace*{0.5cm}
\begin{abstract}
In this paper, we defined the $\mathbb{Z}_q$-linear codes and discussed its various parameters. We constructed $\mathbb{Z}_q$-Simplex code and
$\mathbb{Z}_q$-MacDonald code and found its parameters. We have given a lower and  an upper bounds
of its covering radius for q is an even integer.
\end{abstract}
\vspace*{0.5cm}

\vspace{0.5cm}
\vspace{1.5cm}

\noindent
Corresponding author:\\ 
\\
 \hspace{1cm} 
Dr. C. Durairajan\\
 \\ 
\noindent
Assistant Professor,\\
Department of Mathematics, \\
Bharathidasan University\\
Tiruchirappalli, Tamil Nadu, India, Pin - 620 024\\

\noindent
E-mail: cdurai66@rediffmail.com
\newpage


\section{Introduction}

\quad A code C is a subset of $\mathbb Z_{q}^{n},$ where $\mathbb Z_{q}$ is the  set
of  integer modulo q and n is any positive integer.
 Let $x, y\in\mathbb{Z}_{q}^{n},$ then the distance between $x$ and $y$ is the number of
coordinates in which they differ. It is denoted by $d(x,y).$
Clearly $d(x,y)=wt(x-y),$ the number of non-zero coordinates in $x-y.$ wt(x) is called {\it weight of x}. The minimum distance d of C is
defined by
$$d=\min\{d(x,y)\ \mid x,y\in C \ and\ x\neq y \}.$$ 

The minimum weight of C is $\min\{wt(c) \mid c \in C \text{ and } c\ne 0\}.$
A code of length n cardinality M with minimum distance d over $\mathbb Z_{q}$ is called (n, M, d) q-ary code. For basic
results on coding theory, we refer \cite{mac}.

We know that  $\mathbb Z_{q}$ is a group under addition modulo q. Then $\mathbb Z_{q}^{n}$ is a
group under coordinatewise addition modulo q. A subset C of $\mathbb Z_{q}^{n}$ is said to be a
{\it q-ary code}. If C is a subgroup of $\mathbb Z_{q}^{n},$ then C is called a {\it $\mathbb
Z_{q}$-linear code}.  Some authors are called this code as 
{\it modular code} because $\mathbb Z_{q}^{n}$ is a module over the ring $\mathbb Z_{q}.$ In fact, it is a free $\mathbb Z_{q}$-module.
Since $\mathbb Z_{q}^{n}$ is  a free $\mathbb Z_q$-module, it has a basis. Therefore, every
$\mathbb Z_q$-linear code has a basis. Since $\mathbb Z_{q}$ is finite, it is finite dimension.

Every k dimension $\mathbb Z_{q}$-linear code with length n and minimum distance d is called
{\it $[n,k,d]$ $\mathbb Z_q$-linear code}. A matrix whose rows are a basis elements of the
$\mathbb{Z}_q$-linear code  is called a {\it generator matrix} of C. There are many researchers doing
research on code over finite rings \cite{bu99}, \cite{dgh99},\cite{dghsr99},\cite{dhs99},\cite{el03},
\cite{gupta}, \cite{vsr96}. In the last decade, there are many researchers doing research on codes over $\mathbb Z_{4}$
\cite{aghos99}, \cite{bgl99},  \cite{bsbm97},  \cite{cs93}, \cite{hkcss94}.\\
\hspace*{.7cm}In this correspondence, we concentrate on code over $\mathbb Z_{q}$ where q is even. We constructed some new codes and obtained its
various parameters and its covering radius. In particular, we defined $\mathbb Z_{q}$-Simplex code, $\mathbb Z_{q}$-MacDonald code and studied  its
various parameters. Section $2$ contains basic results for the $\mathbb Z_{q}$-linear codes 
and  we constructed some $\mathbb Z_{q}$-linear code and given its parameters. $\mathbb
Z_{q}$-Simplex code and perfect code are given in section $3$ and
finally, section $4$ we determined the covering radius of these codes and $\mathbb Z_{q}$-MacDonald
code.

\section{$\mathbb Z_q$-Linear code}

\quad  Let C be a $\mathbb Z_{q}$-linear code. If $x,y\in C,$ then $x-y\in C.$
Let us consider the minimum distance of C is
$d = \min\{d(x,y) \mid x,y \in C \text{ and } x \neq y \}.$ Then
$$d=\min\{wt(x-y) \mid x,y \in C \ \text{and } x\neq y \}.$$  
Since C is $\mathbb Z_{q}$-linear code and $x,y \in C$, $x-y \in C.$ Since $x
\neq y,$ $\min\{wt(x-y) \mid x,y \in C \ \text{ and }
x\neq y \} =\min\{wt(c) \mid c \in C \ \text{ and }c \neq 0 \}.$  Thus, we
have
\begin{lemma} In a $\mathbb Z_{q}$-linear code, the minimum distance is the same as the
minimum weight.
\end{lemma}

Let q be an even integer and let $x,y\in \mathbb Z_{q}^{n}$ such that $x_{i},y_{i} \in \{0,
\frac{q}{2}\},$ then $x_{i}\pm y_{i}\in \{0, \frac{q}{2}\}.$
\begin{lemma}
 Let q be an integer even. If $x,y\in \mathbb Z_{q}^{n}$ such that $x_{i},y_{i} \in \{0,
\frac{q}{2}\},$ then the coordinates of $x\pm y$ are either 0 or $\frac{q}{2}.$
\end{lemma}

Now, we construct a new code and discuss its parameters. Let C be an [n, k, d] $\mathbb
Z_q$-linear code. Define
$$D=\{c0c\cdots c + \alpha({\bf 0} 1 {\bf 1}{\bf 2}\cdots {\bf q-1}) \mid \alpha \in  \mathbb Z_q,
c \in C \text{ and } {\bf i}=ii\cdots i \in \mathbb Z_q^n\}.$$
Then, $D=\{c0c\cdots c, c0c\cdots c + {\bf 0} 1 {\bf 1}{\bf 2}\cdots {\bf q-1},  c0c\cdots c +
2({\bf 0} 1 {\bf 1}{\bf 2}\cdots {\bf q-1})\\,\cdots, c0c\cdots c + (q-1)({\bf 0} 1 {\bf 1}{\bf
2}\cdots {\bf q-1})  \mid c \in C \text{ and } \bf{i} \in \mathbb Z_q^n\}.$
Since any $\mathbb{Z}_q$-linear combination of D is again an
element in D, therefore the minimum distance of D is $d(D) = \min\{wt(c0c\cdots c), wt( c0c\cdots c
+ {\bf 0} 1 {\bf 1}{\bf 2}\cdots {\bf q-1} ),\\ wt(c0c\cdots c +
2({\bf 0} 1 {\bf 1}{\bf 2}\cdots {\bf q-1})), \cdots , wt(c0c\cdots c +
(q-1)({\bf 0} 1 {\bf 1}{\bf 2}\cdots {\bf q-1}))
\\ \mid  c \in C \text{ and } \bf{i} \in \mathbb Z_q^n\}.$

Clearly $\min\{wt(c0c\cdots c)  \mid c \in C \& c\ne 0\} \geq qd.$

Let $c\in C.$ Let us take c has $r_{i}\ i's$ where $i = 0,1,2,\cdots, q-1.$ Then for $1 \leq i
\leq q-1,$
$$wt(c+{\bf i})=\sum_{j=0}^{q-1}r_{j} -r_{q-i}.$$
That is $wt(c+{\bf i})=n -r_{q-i}.$
 Therefore
 \begin{eqnarray*}
 wt(c0c\cdots c + {\bf 0} 1 {\bf 1}{\bf 2}\cdots {\bf q-1})& = &wt(c + {\bf 0})+1+ wt(c
+ {\bf 1})+ wt(c+ {\bf 2})\\ && +\cdots + wt(c+{\bf q-1}) \\
&=& n-r_0+1+n-r_{q-1}+n-r_{q-2}+ \cdots \\&& +n-r_1\\
&= &(q-1)n+1
 \end{eqnarray*}
 Similarly, for every  integer i which is relatively prime to q
 \begin{eqnarray*}
 wt((c0c\cdots c) + i({\bf 0} 1 {\bf 1}{\bf 2}\cdots {\bf q-1}))& = &(q-1)n+1
 \end{eqnarray*}

\noindent For other i's
\begin{eqnarray*}
\min\limits_{i\in\mathbb Z_{q}} \{wt(c0c\cdots c + i({\bf 0} 1{\bf 1}{\bf 2}\cdots {\bf q-1}))\} & =& wt(c + {\bf 0})+1\\&&+
wt(c\cdots c+\frac{q}{2}({\bf 1}{\bf 2}\cdots {\bf q-1})) \\
 & =& wt(c + {\bf 0})+1\\&&+ wt(c\cdots c+({\bf \frac{q}{2}}{\bf 0} {\bf \frac{q}{2}}{\bf 0} \cdots
{\bf \frac{q}{2}}{\bf 0}{\bf \frac{q}{2}})) \\
 & =& \frac{q}{2}wt(c + {\bf 0})+1+ \frac{q}{2}wt(c+ {\bf \frac{q}{2}})\\
& =& \frac{q}{2}(n-r_0)+1+ \frac{q}{2}(n-r_{\frac{q}{2}})\\
\min\limits_{i\in\mathbb Z_{q}}\{wt(c0c\cdots c + i({\bf 0} 1{\bf 1}{\bf 2}\cdots {\bf q-1}))\} & =& \frac{q}{2}n+1+ \frac{q}{2}(n-r_0-r_{\frac{q}{2}})\\
\end{eqnarray*}
Hence, $d(D)= \min\{qd,(q-1)n+1,\frac{q}{2} n+1+\frac{q}{2}(n-r_0-r_{\frac{q}{2}})\}.$ Thus, we have
 \begin{theorem} Let C be an $[n, k, d]$ $\mathbb Z_{q}$-linear code, then the\\ 
 $D=\{c0c\cdots c + \alpha({\bf 0} 1 {\bf 1}{\bf 2}\cdots {\bf q-1}) \mid \alpha \in  \mathbb Z_q,
c \in C \text{ and } {\bf i}=ii\cdots i \in \mathbb Z_q^n\}$ is a $[qn+1,k+1,d(D)]$ $\mathbb Z_{q}$-linear code.
\end{theorem}
 If there is a codeword $c\in C$ such that it has only 0 and $\frac{q}{2}$ as coordinates,
then \[
\begin{array}{ccl}
 wt(c0c\cdots c+{\bf0\frac{q}{2}\frac{q}{2}0\frac{q}{2}\cdots 0 \frac{q}{2}})&=&wt(c+0)+1+wt(c+\frac{q}{2})+wt(c\\&&+ 0)+ \cdots+w(c+\frac{q}{2})\\
                                                                             &=&1+r_{\frac{q}{2}}+r_{0}+r_{\frac{q}{2}}+\cdots+r_{0}\\
                                                                             &=&\frac{q}{2}(r_{0}+r_{\frac{q}{2}})+1\\
 wt(c0c\cdots c+{\bf0\frac{q}{2}\frac{q}{2}0\frac{q}{2}\cdots 0 \frac{q}{2}})&=&\frac{q}{2}n+1.
\end{array}\]
Hence, $d(D)=\min\{qd,\frac{q}{2} n+1\}.$  Thus, we have
\begin{corollary}\label{cor1}
If there is a  $c \in C$ such that $c_{i}=0 \text{ or } \frac{q}{2}$ and if $n \leq 2d -1,$ then
$d(D)=\frac{q}{2} n+1.$
\end{corollary}
\section{$\mathbb Z_q$-Simplex codes}
\quad  Let G be a matrix over $\mathbb Z_{q}$ whose columns are one
non-zero element from each 1-dimensional submodule of $\mathbb Z_{q}^{2}.$ Then this matrix is
equivalent to $$G_{2}=\left[\begin{array}{c|c|ccccc}
0&1& 1&2 & \cdots & q-1  \\\hline
1&0& 1& 1 & \cdots &1 \end{array}\right].$$
Clearly $G_{2}$ generates $[q+1, 2, \frac{q}{2}+1]$ code. Inductively, we define
$$G_{k+1}= \left[\begin{array}{c|c|c|c|c|c}
 0 0 \cdots 0 & 1      &1 1 \cdots 1& 2 2\cdots 2 & \cdots & q-1q-1 \cdots q-1 \\\hline
              & 0      &            &             &        &\\
         G_{k}& \vdots & G_{k}      &G_{k}        &\cdots  &G_{k}\\
              & 0      &            &             &        &\\
            \end{array}\right] $$
for $k\geq 2.$ Clearly this $G_{k+1}$ matrix
generates $[n_{k+1}=\frac{q^{k+1}-1}{q-1},\ k+1,\ d]$ code. We call this code as {\it $\mathbb Z_q$- Simplex code}.
This type of k-dimensional code is denoted by $S_{k}(q).$ For simplicity, we denote it by $S_k.$

\begin{theorem}
 $S_{k}(q)$ is  $[n_{k}=\frac{q^{k}-1}{q-1},\ k,\ \frac{q}{2} n_{k-1}+1]$ code.
\end{theorem}
\begin{proof}
We prove this theorem by induction on k. For $k = 2,$ from the generator matrix $G_2,$ it is clear
that $d = \frac{q}{2}+1$ and the theorem is true. Since there is a codeword $c = 0\frac{q}{2}\frac{q}{2}
0\frac{q}{2} \cdots 0 \frac{q}{2} 0 \frac{q}{2} \in S_2$ and 
$n = q + 1 \leq 2(\frac{q}{2} +1) - 1 = 2d-1 ,$ by Corollary \ref{cor1} implies $d(S_3) =
\frac{q}{2}n_2+1  $ and hence the $S_3$ is $[n_3=\frac{q^3 -1}{q-1}, 3, \frac{q}{2}n_2+1 ]$
code. Since $c0c\cdots c + \frac{q}{2}({\bf 0} 1 {\bf 1}{\bf 2}\cdots {\bf q-1}) \in S_3$ whose
coordinates are either 0 or $\frac{q}{2}$ and satisfies the conditions of the Corollary \ref{cor1},
therefore $d(S_4) = \frac{q}{2}n_3+1 $ and hence the
$S_4$ is $[n_4=\frac{q^4 -1}{q-1}, 4, \frac{q}{2}n_3+1 ]$ code. By induction we can assume that
this theorem is true for all less
than k. That is, there is a code $c \in S_{k-1}$ whose coordinates are either 0 or $\frac{q}{2}$ and $n_{k-1}
\leq 2d_{k-1} - 1.$ By Corollary \ref{cor1}, $d_k = \frac{q}{2} n_{k-1} +1.$
Therefore $S_{k}(q)$ is an $[\frac{q^{k}-1}{q-1},\ k,\ \frac{q}{2} n_{k-1}+1] \ \mathbb Z_q$-linear code.
Thus we proved.
\end{proof}

This code seems to be  a Simplex code over finite field  but not because the parameters differ. From the matrix $G_k$, no two
columns are linearly dependent. Therefore the minimum distance of its dual is greater than or
equal to 3. Since in the first block of the matrix $G_k,$ there are three columns
 whose transpose matrix are $(0, 0, \cdots , 0, 1), (0, 0, \cdots ,0, 1, 0), (0, \\0, \cdots ,0, 1,
1).$ These are linearly dependent. Therefore, the minimum distance of the dual code is less than
or equal to 3. Hence the dual of $S_k$ is $[n_k=\frac{q^k-1}{q-1}, n_k-k,3]$   $\mathbb
Z_{q}$-linear code. One can check easily that the spheres of radius 1 around codewords cover
the whole space $\mathbb Z_{q}^{n}$ and the spheres are disjoint. Hence the dual of this code is
perfect. This is, this code is equivalent to the q-ary Hamming codes because any code with these parameters
$(n=\frac{q^k-1}{q-1},q^{n-k},3)$ is equivalent to the q-ary Hamming code\cite{mac}.

\section{Covering radius}
\quad  The {\em covering radius} of a code $C$ over  $\mathbb Z_{q}$ with respect to the Hamming distance
$d$ is given by
$$R(C)= \max_{u \in \mathbb Z_{q}^{n}}\left\{\min_{c \in C}\left\{d(u,c)\right\}\right\}.$$
It is easy to see that $R(C)$ is the least positive integer $r$ such that
$${\mathbb Z_{q}^{n}}= \cup_{c \in  C} S_{r}(c)$$ where
$$ S_{r}(u)=\left\{ v \in \mathbb Z_{q}^{n}\} \mid d(u,v) \leq r\right\}$$ for
any $u \in { \mathbb Z_{q}^{n} }.$

\begin{proposition}\cite{survey}\label{append}
If appending( puncturing) r number of columns in a code C, then the covering radius of C is
increased( decreased ) by r.
\end{proposition}

\begin{proposition}\label{prop} \cite{mat}
If $ C_0$ and $ C_1$ are codes over $\mathbb Z_{q}^{n}$  generated by matrices $G_0$ and $G_1$
respectively and if $ C$ is the code generated by
\[
G =  \left( \begin{array}{c|c}
0 & G_1 \\\hline
G_0 & A
\end{array}\right),
\]
then $r( C) \leq r( C_0) + r( C_1)$ and  the covering radius of $C$ 
satisfy the following
\[
r( C) \geq r( C_0) + r( C_1).
\]
  Since the covering radius of C generated by
\[
G =  \left( \begin{array}{c|c}
0 & G_1 \\\hline
G_0 & A
\end{array}\right),
\] is greater than or equal to $r(C_{0})+r(C')$
where $C_{0}$ and $C'$ are codes generated by   $ \left[ \begin{array}{c}
 0 \\\hline
 G_0
\end{array}\right] =  \left[ \begin{array}{c}
 G_0
\end{array}\right]$ and $ \left[ \begin{array}{c}
 G_1 \\\hline
 A
\end{array}\right],$ respectively, this implies $r( C) \geq r( C_0) + r( C_1)$ because $C_{1}$ is a subcode of the
code 
 $ C'.$
\end{proposition}

A $q$-ary repetition code $\ C$ over a finite field $\mathbb F_q$ with q elements is
an $[n,1,n]$ linear code. The covering
radius of $\ C$ is $\lfloor \frac{n(q-1)}{q} \rfloor$ \cite{duraithesis96}. For basic results
on covering radius, we refer to \cite{survey}, \cite{further}. Now, we consider the
repetition
code over $\mathbb Z_{q}$. There are two types of repetition codes.
\begin{enumerate}
 \item[Type I.] Unit repetition  code  generated by
$G_u=[\overbrace{u u \ldots u}^{n}]$ where $u$ is an unit element of $\mathbb Z_{q}.$  This matrix
generates $ C_u$ is $[n,1,n]$ $\mathbb{Z}_q$-linear code. That is, $C_{u}$ is (n, q, n) q-ary repetition code. We call this as {\it unit repetition code.}
\item[Type II.]   Zero divisor repetition code is generated by the
matrix\\  $G_v=[\overbrace{ v v \ldots v}^{n} ]$ where $v$ is a zero divisor in $\mathbb Z_{q}.$ That
is, $v$ is not a relatively prime to q. This is an  $(n,\frac{q}{v},n)$ code over $\mathbb{Z}_q.$
This
 code is denoted by $C_v.$ This code is called {\it zero divisor repetition code}.
\end{enumerate}

\hspace{.7cm}With respect to the Hamming distance the covering radius
 of $C_1\ is \lfloor \frac{n(q-1)}{q}\rfloor$ \cite{duraithesis96} but clearly the covering radius
of $C_v$ is n because code symbols appear in this code are zero divisors only. Thus, we
have
\begin{theorem}
$R( C_v)=n$ and $R( C_u)=\lfloor\frac{(q-1)n}{q}\rfloor.$
\end{theorem}
Let $\phi(q) = \# \{ i \mid 1 \leq i \textless \ q \ \& \ (i,q) = 1\}$ be the Euler $\phi$-function. Let
$U = \{  i\in\mathbb{Z} \mid 1 \leq i \textless \ q \ \&\ (i,q) = 1\}$ be the set of all units in $\mathbb Z_{q}$ and let
$O = \mathbb Z_{q} \setminus U$ be the set which contains all zero divisors and 0.
Let C be a $\mathbb Z_{q}$-linear code generated by the matrix
$$[ \overbrace{1 1 \ldots 1 }^{n}\overbrace{2 2 \ldots 2 }^{n} \cdots\overbrace{q-1 q-1 \ldots q-1 }^{n}],$$ then this code is equivalent to a code whose
generator matrix is $[u_1u_1\cdots \\ u_1 u_2u_2\cdots u_2 \cdots
u_{\phi(q)}u_{\phi(q)}\cdots u_{\phi(q)} o_1o_1 \cdots o_1 o_2o_2 \cdots o_2  \cdots
o_ro_r \cdots o_r]$\\ where $r = q -1 -\phi(q).$
Let A be a code equivalent to the unit repetition code of length $\phi(q)n$ generated by
$[u_1u_1\cdots u_1u_2u_2\cdots u_2 \cdots u_{\phi(q)}u_{\phi(q)}\cdots \\ u_{\phi(q)}],$ then by the
above theorem, $R(A) =  \lfloor\frac{(q-1)\phi(q)n}{q}\rfloor.$ Let B be a code equivalent to the
zero divisor repetition code  of length $(q-1-\phi(q))n$ generated by $[o_1o_1 \cdots o_1 o_2o_2
\cdots o_2 \cdots o_ro_r \cdots o_r],$ then by the above theorem,
$R(B) = (q-1-\phi(q))n.$ By Proposition \ref{prop}, $R(C) \geq
\lfloor\frac{(q-1)\phi(q)n}{q}\rfloor+(q-1-\phi(q))n.$

 Without loss of generality we can assume that the generator matrix of A
as $[111 \cdots 1].$  Since $R(A) = \lfloor\frac{(q-1)\phi(q)n}{q}\rfloor$ and C is obtained by
appending some $(q-1-\phi(q))n$ columns to A, by Proposition \ref{append} the covering radius of C
is increased by at most $(q-1-\phi(q))n.$ Therefore, $R(C) \leq
\lfloor\frac{(q-1)\phi(q)n}{q}\rfloor +(q-1-\phi(q))n.$
Thus, we have
\begin{theorem}{\label{repetition}}
 Let C be a $\mathbb Z_{q}$-linear code generated by the matrix
$[ \overbrace{1 1 \ldots 1 }^{n}\\ \overbrace{2 2 \ldots 2 }^{n} \cdots\overbrace{q-1 q-1 \ldots q-1 }^{n}].$ Then C is a $[(q-1)n,1, \frac{q}{2}n]$ $\mathbb
Z_{q}$-linear code with $R(C)=\lfloor\frac{(q-1)\phi(q)n}{q}\rfloor+(q-1-\phi(q))n.$
\end{theorem}

Now, we see the covering radius of $ \mathbb Z_{q}$-Simplex  code. The covering radius
of Simplex codes and MacDonald codes over finite field and finite rings were discussed in
\cite{duraithesis96}, \cite{gupta}.
\begin{theorem} For $k\geq 2,$ $R(S_{k+1}) \leq
\frac{(k-1)(q-1)\phi(q)+(q^2-q-\phi(q))(q^{k+1}-q^2)}{q(q-1)^2} \\ \hspace*{6cm}+ R(S_{2})$
\end{theorem}
\begin{proof} For $k \geq 2,$  $S_{k+1}$ is $[n_{k+1} = \frac{q^{k+1} -1}{q-1}, k+1,
\frac{q}{2}n_k+1]$ $\mathbb
Z_{q}$-linear code. By Proposition \ref{prop} and Theorem \ref{repetition}  give
 \begin{eqnarray*}
R(S_{k+1})& \leq & (1 + \lfloor\frac{(q-1)\phi(q)n_k}{q}\rfloor+(q-1-\phi(q))n_k)+R(S_{k})\\
           & \leq & (1 + \frac{(q-1)\phi(q)n_k}{q}+(q-1-\phi(q))n_k)+R(S_{k})\\
           & \leq & 1 + \frac{q^2-q-\phi(q)}{q}n_k+R(S_{k})\\
R(S_{k+1})& \leq & (1 + \frac{q^2-q-\phi(q)}{q}n_k)+R(S_{k})\\
\end{eqnarray*}
 This implies
 $$R(S_{k}) \leq  (1 + \frac{q^2-q-\phi(q)}{q}n_{k-1})+R(S_{k-1}).$$
 Combining these two, we get
  $$R(S_{k+1})\leq  (1 + \frac{q^2-q-\phi(q)}{q}n_k)+ (1 +
\frac{q^2-q-\phi(q)}{q}n_{k-1})+R(S_{k-1})$$
Similarly, if we continue, we get

  $R(S_{k+1})\leq  (1 + \frac{q^2-q-\phi(q)}{q}n_k)+ (1 +
\frac{q^2-q-\phi(q)}{q}n_{k-1})+ \cdots +\\\hspace*{2.4cm}(1 + \frac{q^2-q-\phi(q)}{q}n_2)+ R(S_{2})$\\
 Since $n_k=\frac{q^{k}-1}{q-1}, \text{ for } k \geq 2,$ therefore
 \begin{eqnarray*}
R(S_{k+1}) & \leq & (k-1)+\frac{q^2-q-\phi(q)}{q}\left(\frac{q^{k}-1}{q-1}+ \frac{q^{k-1}-1}{q-1}+\cdots +
                     \frac{q^{2}-1}{q-1}\right)+ R(S_{2})\\
           & \leq & (k-1)+\frac{q^2-q-\phi(q)}{q}\left(\frac{q^{k}+q^{k-1} + \cdots + q^{2}-(k-1)}{q-1}\right)+
                        R(S_{2})\\
           & \leq & \frac{(k-1)\phi(q)+(q^2-q-\phi(q))((q^{k+1}-1)/(q-1)-(q+1))}{q(q-1)} +  R(S_{2})\\
R(S_{k+1})& \leq &\frac{(k-1)(q-1)\phi(q)+(q^2-q-\phi(q))(q^{k+1}-q^2)}{q(q-1)^2} +  R(S_{2})\\
 \end{eqnarray*}
  Hence proved.
 \end{proof}
 In particular, for  $q=4,$
  $R(S_{k+1})\leq
\frac{5.4^{k+1}+3k-29}{18} \text{ for } k \geq 2$ because of simple calculation $R(S_{2})=3.$

 Now, we can define a new code which is similar to the $\mathbb Z_{q}$-MacDonald  code. Let
 $$G_{k,u} = \left(\begin{matrix}
                   G_k \setminus \left(\begin{matrix}
0\\
G_u
\end{matrix}\right)
\end{matrix} \right)$$
 for $2 \leq u \leq k-1$ where 0 is a $(k-u)\times\frac{q^u-1}{q-1}$ zero matrix and $\left(\begin{matrix}
                   A \setminus B
\end{matrix} \right)$ is a matrix obtained from the matrix A by removing the matrix B. The code
generated by $G_{k,u}$ is called {\it $\mathbb Z_{q}$-MacDonald} code. It is denoted by $M_{k,u}.$ The Quaternary
MacDonald codes were discussed in \cite{cg03}.
\begin{theorem} For $k\geq 2$ and $0\leq u \leq k,$\\
$$R(M_{k+1,u}) \leq \frac{(k-r+1)(q-1)\phi(q)+(q^2-q-\phi(q))q^r(q^{k-r+1}-1)}{q(q-1)^2}$$ \hspace*{3cm}$+ R(M_{r,u}),\ \mbox{for}\ u\leq  r \leq k.$
\end{theorem}
\begin{proof}
 By using, Proposition \ref{prop}, we get
 \begin{eqnarray*}
R(M_{k+1,u})& \leq & (1 + \lfloor\frac{(q-1)\phi(q)n_k}{q}\rfloor+(q-1-\phi(q))n_k)+R(M_{k,u})\\
             & \leq & (1 + \frac{(q-1)\phi(q)n_k}{q}+(q-1-\phi(q))n_k)+R(M_{k,u})\\
             & \leq & 1 + \frac{q^2-q-\phi(q)}{q}n_k+R(M_{k,u})\\
R(M_{k+1,u})& \leq & (1 + \frac{q^2-q-\phi(q)}{q}n_k)+R(M_{k,u})\\
\end{eqnarray*}
 This implies
 $R(M_{k,u}) \leq  (1 + \frac{q^2-q-\phi(q)}{q}n_{k-1})+R(M_{k-1,u}).$
 Combining these two, we get
  $$R(M_{k+1,u})\leq   (1 + \frac{q^2-q-\phi(q)}{q}n_k)+ (1 + \frac{q^2-q-\phi(q)}{q}n_{k-1})+
R(M_{k-1,u})$$
Similarly, if we continue, we get
  \begin{eqnarray*}
R(M_{k+1,u})&\leq &(1 + \frac{q^2-q-\phi(q)}{q}n_k)+ (1 +
\frac{q^2-q-\phi(q)}{q}n_{k-1})\\
 && +\cdots + (1 + \frac{q^2-q-\phi(q)}{q}n_{r})+R(M_{r,u}).
  \end{eqnarray*}
 Since $n_k=\frac{q^{k}-1}{q-1}, \text{ for } k \geq 2,$ therefore
 \begin{eqnarray*}
R(M_{k+1,u}) & \leq &
(k-r+1)+\frac{q^2-q-\phi(q)}{q}\left(\frac{q^{k}-1}{q-1}+ \frac{q^{k-1}-1}{q-1}+\cdots +
\frac{q^{r}-1}{q-1}\right)\\&&+ R(M_{r,u})\\
& \leq &
(k-r+1)+\frac{q^2-q-\phi(q)}{q}\left(\frac{q^{k}+q^{k-1} + \cdots +
q^{r}-(k-r+1)}{q-1}\right)\\&&+ R(M_{r,u})\\
 & \leq &
\frac{(k-r+1)\phi(q)+(q^2-q-\phi(q))q^r(q^{k-r}+q^{k-r-1} + \cdots + 1)}{q(q-1)}\\&&+  R(M_{r,u})\\
 & \leq &
\frac{(k-r+1)\phi(q)+(q^2-q-\phi(q))q^r(q^{k-r+1}-1)/(q-1)}{q(q-1)} \\&& + R(M_{r,u})\\
R(M_{k+1,u})& \leq &
\frac{(k-r+1)(q-1)\phi(q)+(q^2-q-\phi(q))q^r(q^{k-r+1}-1)}{q(q-1)^2} \\&&+ R(M_{r,u}).
 \end{eqnarray*}
 \end{proof}

If $u = k,$ then $$R(M_{k+1,k}) \leq
\lfloor\frac{(q-1)\phi(q)n_k}{q}\rfloor+(q-1-\phi(q))n_k +1 \text{ for } k \geq 2.$$
In the above theorem, if we replace r by $u+1,$ we get
\begin{eqnarray*}
 R(M_{k+1,u})& \leq &
\frac{(k-u)(q-1)\phi(q)+(q^2-q-\phi(q))q^{u+1}(q^{k-u}-1)}{q(q-1)^2} \\
&&+\frac{(q-1)\phi(q)n_u}{q}+(q-1-\phi(q))n_u +1 \text{ for } u \geq 2.
\end{eqnarray*}
Thus, we have
\begin{corollary} For $k \geq 2 \text{ and } k \geq u \geq 2,$
 \begin{eqnarray*}
 R(M_{k+1,u})& \leq &
\frac{(k-u)(q-1)\phi(q)+(q^2-q-\phi(q))q^{u+1}(q^{k-u}-1)}{q(q-1)^2} \\
&&+\frac{(q-1)\phi(q)n_u}{q}+(q-1-\phi(q))n_u +1.
\end{eqnarray*}
\end{corollary}

\end{document}